\documentclass[12pt]{article}
\usepackage{amsmath}
\DeclareMathOperator*{\argmin}{arg\,min}

\usepackage{graphicx,psfrag,epsf}
\usepackage{enumerate}
\usepackage{natbib}
\usepackage{url} % not crucial - just used below for the URL 
\usepackage{amsfonts}
\usepackage{amsthm}
\usepackage{verbatim} % used for comment
\usepackage{color}
\newcommand{\blind}{1}

\def\T{{ \mathrm{\scriptscriptstyle T} }}

\newtheorem{definition}{Definition}
\newtheorem{theorem}{Theorem}
\newtheorem{remark}{Remark}
\newtheorem{proposition}{Proposition}
\newtheorem{lemma}{Lemma}

\begin{document}

\date{}
\def\spacingset#1{\renewcommand{\baselinestretch}%
{#1}\small\normalsize} \spacingset{1}

%%%%%%%%%%%%%%%%%%%%%%%%%%%%%%%%%%%%%%%%%%%%%%%%%%%%%%%%%%%%%%%%%%%%%%%%%%%%%%

\if1\blind
{
  \title{\bf Matrix Linear Discriminant Analysis}
  \author{Wei Hu, Weining Shen
  \hspace{.2cm}\\
  Department of Statistics, University of California, Irvine\\\\
  Hua Zhou\\
  Department of Biostatistics, University of California, Los Angeles\\\\
    Dehan Kong \footnote{
Address for correspondence:
 Dehan Kong, Ph.D., E-mail:
kongdehan@utstat.toronto.edu}\\
    Department of Statistical Sciences, University of Toronto
    }
  \maketitle
} \fi

\if0\blind
{
  \bigskip
  \bigskip
  \bigskip
  \begin{center}
    {\LARGE\bf Matrix Linear Discriminant Analysis}
\end{center}
  \medskip
} \fi

\bigskip
\begin{abstract}
We propose a novel linear discriminant analysis approach for the classification of high-dimensional matrix-valued data that commonly arises from imaging studies.  Motivated by the equivalence of the conventional linear discriminant analysis and the ordinary least squares, we consider an efficient nuclear norm penalized regression that encourages a low-rank structure. Theoretical properties including a non-asymptotic risk bound and a rank consistency result are established. Simulation studies and an application to electroencephalography data show the superior performance of the proposed method over the existing approaches. 
\end{abstract}

\noindent%
\vfill

\newpage
\spacingset{1.45} % DON'T change the spacing!
\section{Introduction}
\label{sec:intro}

Modern technologies have generated a large number of datasets that possess a matrix structure for classification purpose. For example, in neuropsychiatric disease studies, it is often of interest to evaluate the prediction accuracy of prognostic biomarkers by relating two-dimensional imaging predictors, e.g., electroencephalography (EEG) and magnetoencephalography, to clinical outcomes such as diagnostic status \citep{Mu2011}. In this paper, we focus on extending one of the most commonly used classification methods, Fisher linear discriminant analysis (LDA) to matrix-valued predictors. Progress has been made in recent years on developing sparse LDA using $\ell_1$-regularization \citep{tibshirani1996}, including \citet{Shao2011, Fan2012, mai2012}. However, all these methods only deal with vector-valued covariates; and it remains challenging to accommodate the matrix structure. Naively transforming the matrix data into a high-dimensional vector will result in unsatisfactory results for several reasons. First, vectorization destroys the structural information within the matrix such as shapes and spatial correlations. Second, turning a $p \times q$ matrix into a $pq \times 1$ vector generates unmanageably high dimensionality. E.g., estimating the population precision matrix for LDA can be troublesome if $pq \gg n$. Third, $\ell_1$-regularization does not necessarily work well because the underlying two-dimensional signals are usually approximately low-rank rather than $\ell_0$-sparse. 

Recently, there are some development of regression methods for matrix data. \citet{chen2013reduced} invented an adaptive nuclear norm penalization approach for low-rank matrix approximation. \citet{zhou2014} proposed a class of regularized matrix regression methods based on spectral regularization. \citet*{wang2017} developed a generalized scalar-on-image regression model via total variation. \citet{kong2018l2rm} proposed a low-rank linear regression model with high-dimensional matrix response and high dimensional scalar covariates, while \citet{hu2019nonparametric} developed a nonparametric matrix response regression model. 

In this paper, we propose a new matrix LDA approach by building on the equivalence between the classical LDA and the ordinary least squares. We formulate the binary classification as a nuclear norm penalized least squares problem, which efficiently exploits the low rank structure of the two-dimensional discriminant direction matrix. The involved optimization is amenable to the accelerated proximal gradient method. Although our problem is formulated as a penalized regression problem, a fundamental difference is that the covariates ${\bf X}_i$ and the residuals $\epsilon_i$ are no longer independent in our case. This requires extra effort for developing the risk bound and rank consistency result. The risk bound is explicit in terms of the rank of the image, image size, sample size, and the eigenvalues of the covariance matrix for the image covariates. This result also implies estimation consistency provided the $p \times q$ image satisfies $\max(p,q) = o(n/\log^3 n)$. Under stronger conditions, we show that the rank of the coefficient matrix can be consistently estimated as well. The proof is based on exploiting the spectral norm of random matrices with mixture-of-Gaussian components and extending the results in \citet{Bach2008} to allow diverging matrix dimensions. Finally, we prove that our method enjoys classification error consistency. 

It is worth noting that the 2D image classification problem has been studied by \citet{zhong2015matrix}, where they proposed a penalized matrix discriminant analysis method (PMDA) that projects the matrix coefficient into row space and column space separately. Those two projections are then estimated iteratively and integrated together for classification. Compared with PMDA, we make the following contributions. First, the rank of the PMDA is set as one because of the separability assumption, while we allow the rank of the direction matrix to take general positive integer values and the rank can then be selected by a data driven procedure. Our rank assumption is more flexible in practice and hence often leads to a lower mis-classification error in the numerical studies. Second, our method adopts a direct estimation approach by solving a nuclear norm penalized  regression problem, which is computationally much faster compared with PMDA, where the estimation involves an iterative procedure for calculating the inverse of covariance matrices during each iteration. Third, our method can handle the high-dimensional data when image dimensions $p $ and $ q $ are much larger than the sample size, which is the case for many applications; while PMDA cannot handle the case when $ p+q>n$. Finally, we have provided theoretical guarantee for our estimator when $ p $ and $ q $ diverge with $ n $. In particular, we have developed an non-asymptotic error bound for the estimated LDA direction, as well as results on rank consistency and classification error consistency. These results are stronger compared with the root-$n$ consistency of the LDA direction in \citet{zhong2015matrix}, where both $ p $ and $ q $ are assumed to be fixed.

\section{Method}
We first define some useful notations. Let $\text{vec}(\cdot)$ be a vectorization operator, which stacks the entries of a matrix into a column vector. The inner product between two matrices of same size is defined as $\langle {\bf M, N} \rangle=\text{tr}({\bf M}^{\T} {\bf N})=\langle \text{vec}({\bf M}), \text{vec}({\bf N}) \rangle$.

Consider a binary classification problem, where $\bf X$ is a two-dimensional image covariate with dimension $p\times q$ and $G = 1, 2$ denotes the class labels. The LDA assumes that $\text{vec} ({\bf X}) \mid G = g\sim N({\bf \mu}_g, {\bf \Sigma})$, $\text{pr}(G=1) = \pi_1$, and
$\text{pr}(G=2) = \pi_2$.  Suppose we have $n$ subjects with $n_1$ subjects belonging to class 1 and $n_2=n-n_1$ subjects to class 2.  It is well known that LDA is connected to the linear regression with the class labels as responses \citep{duda2012pattern, mika2002}. When $ pq<n $, the classical LDA is equivalent to solving
\begin{align} \label{OLS}
(\hat\beta_0^{\text{ols}}, \hat{\bf B}^{\text{ols}} ) = \argmin_{\beta_0,B}\sum_{i=1}^n\Big (y_i - \beta_0 - \langle {\bf{X}_i, B}\rangle \Big)^2,
\end{align}
where ${\mathbf X}_i$ is the image covariate from subject $i$, $\mathbf B$ is the coefficient matrix for the image covariate and it represents the direction of the linear discriminant classifier, $\beta_0$ is the intercept, and the response $y_i = -n/n_1$ if subject $i$ is in class 1, and $y_i = n/n_2$ if subject $i$ is in class 2. Although this connection gives the exact LDA direction when $ pq < n $, it has two potential drawbacks. First, when $ pq>n $, the equivalence between Fisher LDA and (\ref{OLS}) is lost because of the non-uniqueness of solution. Second, the formulation (\ref{OLS}) does not incorporate the 2D image structure when estimating the direction because $  \langle {\bf X}_i, {\bf B} \rangle=\langle \text{vec}({\bf X}_i), \text{vec}({\bf B}) \rangle$. These motivate us to consider a penalized version of (\ref{OLS}) as follows
\begin{align} \label{POLS1}
(\hat\beta_0, \hat {\bf B}) = \argmin_{\beta_0, {\bf B}}\frac{1}{2n} \sum_{i=1}^n \Big(y_i - \beta_0 - \langle{\bf X}_i,{\bf B} \rangle \Big)^2 + \omega_n\|{\bf B}\|_\ast, 
\end{align}
where the nuclear norm $ \|{\bf B}\|_{\ast}=\sum_{j}\sigma_j({\bf B}) $ and  $ \sigma_j({\bf B}) $s are the singular values of the matrix $ {\bf B} $. The nuclear norm $ \|{\bf B}\|_\ast $ plays an important role because it imposes a low rank structure in the estimated direction $ \hat {\bf B} $.  An alternative choice is to add a Lasso type penalty, i.e. $ \omega_n\|{\bf B}\|_{1,1}=\omega_n\sum_{j=1}^p \sum_{k=1}^q |b_{jk}|$, where $ b_{jk} $ is the $jk$-th element of $ {\bf B} $. However, the Lasso type penalty can only identify at most $ n $ nonzero components, and for most cases in imaging studies, the signal is usually not that sparse. More importantly, the Lasso type of penalty ignores the matrix structure because it is equivalent to vectorizing the array and applying sparse LDA. Once $ \hat {\bf B} $ from (\ref{POLS1}) is obtained, a naive classification rule will assign the $i$-th subject to class 2 if $  \langle {\bf X}_i, \hat {\bf B} \rangle  + \hat\beta_0 > 0 $. However, it can be shown that the intercept $  \hat\beta_0 $ obtained from (\ref{POLS1}) is not optimal. Instead, we use the optimal intercept $\tilde\beta_0$ that minimizes the training error after obtaining $\hat {\bf B}$. \citet{mai2012} showed that the intercept of LDA actually has a closed form. Their derivations can be easily applied to our case. In particular, if $(\hat\mu_2 - \hat\mu_1)^{\T}\text{vec}(\hat {\bf B}) >0$, then
\begin{eqnarray}\label{optimalintercept}
\tilde\beta_0=-(\hat\mu_1 +\hat\mu_2)^{\T}\text{vec}(\hat {\bf B})/2 + \text{vec}(\hat {\bf B})^{\T}\hat{\bf \Sigma} \text{vec}(\hat {\bf B})\{(\hat\mu_2 - \hat\mu_1)^{\T}\text{vec}(\hat {\bf B})\}^{-1}\log(n_2/n_1),
\end{eqnarray} 
where $\hat\mu_g$ is the sample mean for subjects in class $g$ and $\hat\Sigma$ is the estimated covariance matrix.   If $(\hat\mu_2 - \hat\mu_1)^{\T}\text{vec}(\hat {\bf B}) <0$, we can plug $ -\hat {\bf B} $ into (\ref{optimalintercept}) to obtain the optimal intercept $ \tilde\beta_0 $. The optimal classification rule is to assign the $i$-th subject to class 2 if $ \langle {\bf X}_i, \hat {\bf B} \rangle  + \tilde\beta_0 > 0$.

For any fixed $\omega_n$, the optimization problem in \eqref{POLS1} can be solved using the accelerated proximal gradient method \citep{nesterov1983, beck2009}. \citet*{zhou2014} studied the algorithm for the nuclear norm regularized matrix regression.  As we know, nuclear norm is not differentiable. Fortunately, its subderivative $\partial \|.\|_\ast$ exists. Therefore \eqref{POLS1} has local minima $(\hat\beta_0, \hat {\bf B})$ if and only if $0 \in -\frac{1}{n}\sum_{i=1}^n {\bf X}_i\epsilon_i + \omega_n\partial \|\hat {\bf B}\|_\ast.$ Thanks to the convexity of nuclear norm, the local minima is global as well. Based on these facts, singular value thresholding method for nuclear norm regularization was deployed for building blocks of the Nesterov's method. Compared with classical gradient decent method with convergence of $O(t^{-1})$, where $t$ denotes the number of iteration, Nesterov's  accelerated gradient decent method achieves convergence rate of $O(t^{-2})$. It differs from traditional algorithms by utilizing the estimators from previous two iterations to generate the next estimator. For computational algorithm, we use the {\tt matrix\_sparsereg} function in the Matlab TensorReg Toolbox (\url{https://hua-zhou.github.io/TensorReg/}) for solving nuclear norm penalized matrix regression. It implements an optimal Nesterov acceleration of the proximal gradient algorithm. Actually one contribution of our paper is to link matrix LDA to regularized matrix regression so that the computational machinery developed for the latter can be applied to matrix LDA problems. For tuning of the $ \omega_n $, we adopt the \textsc{bic} derived by \citet*{zhou2014} under the nuclear norm regularized matrix regression framework. 

\section{Theory}
In this section we discuss the theoretical properties of  the proposed regularization estimator. Denote the residuals $\epsilon_i=y_i - \beta_0 - \langle{\bf X}_i,{\bf B} \rangle$ and the true coefficient matrix by ${\bf B}_0$. {By the equivalence between LDA direction and least squares, we know $\text{vec}({\bf B}_0)$ can be written as $c{\bf \Sigma}^{-1}(\mu_2 - \mu_1)$ for some positive constant $c$.} Consider the singular value decomposition ${\bf B}_0 = {\bf U}_0\text{Diag}(S_0){\bf V}_0^{\T}$ with ${\bf U}_0 \in \mathrm{R}^{p \times r}$ and ${\bf V}_0 \in  \mathrm{R}^{q \times r}$. Let ${\bf U}_{0\perp} \in \mathrm{R}^{p \times (p-r)} $ and ${\bf V}_{0\perp} \in \mathrm{R}^{q \times (q-r)}$ be (arbitrary) orthogonal complements of ${\bf U}_0$ and ${\bf V}_0$, respectively. We make the following assumptions. 
\begin{itemize}
\item [(A1)] We assume that the second-order moment of the covariate ${\bf X}$, $\text{E}(\text{vec}({\bf X})\text{vec}({\bf X})^{\T})$, denoted by ${\bf \Sigma}_{xx}$, satisfies $\lambda_l \leq \lambda_{\min}({\bf \Sigma}_{xx}) \leq \lambda_{\max}({\bf \Sigma}_{xx}) \leq \lambda_u$, where $\lambda_{\min}({\bf \Sigma}_{xx})$ and $\lambda_{\max}({\bf \Sigma}_{xx})$ are the smallest and largest eigenvalues of ${\bf \Sigma}_{xx}$, respectively, and $\lambda_l, \lambda_u$ are some positive constants. 
%\item [(A2)] Assume $r = \text{rank}(B_0)$ is fixed but unknown. 
\item [(A2)] Let $ r = \text{rank}({\bf B}_0)$ be the unknown rank of the true coefficient matrix ${\bf B}_0$. Define ${\bf \Lambda} \in \mathrm{R}^{(p-r)\times(q-r)}$ as
\begin{align*}
\text{vec}({\bf \Lambda}) = \{({\bf V}_{0\perp}\otimes {\bf U}_{0\perp})^{\T}{\bf \Sigma}^{-1}({\bf V}_{0\perp}\otimes {\bf U}_{0\perp})\}^{-1}
\{({\bf V}_{0\perp}\otimes {\bf U}_{0\perp})^{\T}{\bf \Sigma}^{-1}({\bf V}_0\otimes {\bf U}_0)\text{vec}({\bf I})\}.
\end{align*}
We assume its spectral norm $\|{\bf \Lambda}\|_2 < 1$.
\item [(A3)] Assume the quantities $\omega_n$,  $\{\min(p,q)\}^{1/2}n^{-1/2}\omega_n^{-1}$, $ \min(p,q)n^{-1/2}$, $\omega_np^{1/2}q^{1/2}\min(p,q)$ tend  to $0$ as $n \rightarrow \infty$.
\item [(A4)] There exists a positive constant $C_{\mu}$ such that $\| {\bf \mu}_2 - \mu_1 \|_2 \leq C_{\mu} (\sqrt{p} + \sqrt{q})$.
\end{itemize}
Condition (A1) requires bounded eigenvalues for the covariance matrix of the vectored covariate, which is standard in the literature. Condition (A2) is similar with the strict consistency condition in \citet{Bach2008}. It is needed to establish rank consistency. This condition extends the classical strong irrepresentable condition in \citet{Zhao2006}, which is commonly used for proving model selection consistency of Lasso. The major difference between our Assumption (A2) and the similar assumption in \citet{Bach2008} is that the number of parameters is fixed in \citet{Bach2008} while in our case the number is diverging with $n$. Therefore we will need to assume that the regularization parameter $\omega_n$ decays slower than the one in \citet{Bach2008}. Condition (A3) puts more requirement on the order of $p,q$, and $w_n$ in order to obtain consistent rank estimation in addition to consistent coefficient estimation. This is expected since rank estimation consistency is usually not implied by parameter estimation consistency. Condition (A4) can be viewed as a sparsity assumption on $B_0$. Recall the solution (the slope) to classical LDA problem with vector covariates depends on the term $\mu_2 - \mu_1$. This assumption essentially implies that there are at most $O(\max(p,q))$ number of $O(1)$ elements in the true coefficient matrix ${\bf B}_0$ given the rank of ${\bf B}_0$ is fixed. 

Next, we briefly review two important concepts, namely decomposable regularizer and strong convex loss function, proposed by \citet{Nega2012} and highlight their connection to the risk bound property for our estimator.  
\begin{definition}
A regularizer $R(\cdot)$ is decomposable with respect to a given pair of subspaces $(M, N)$ where $M  \subseteq N^\perp$ if
\[
R(u + v) = R(u) + R(v) \text{\quad for all  } u\in M, v\in N.
\]
\end{definition}

In our setting, $R(\cdot)$ is the nuclear norm. Considering a matrix ${\bf B} \in \mathcal R^{p\times q}$ to be estimated, we observe that nuclear norm is decomposable given a pair of subspaces:
\[
M({\bf U}, {\bf V}) := \{{\bf B} \in \mathcal R^{p\times q} \mid \text{row}({\bf B}) \subseteq {\bf V}, \text{ col}({\bf B}) \subseteq {\bf U}\},
\]
\[
N({\bf U}, {\bf V}) := \{{\bf B} \in \mathcal R^{p\times q} \mid \text{row}({\bf B}) \subseteq {\bf V}^{\perp}, \text{ col}({\bf B}) \subseteq {\bf U}^{\perp}\},
\]
where ${\bf U}, {\bf V}$ represent $B$'s left and right singular { vectors}. For any pair of matrices ${\bf B}_1 \in M$ and ${\bf B}_2 \in N$, the inner product of ${\bf B}_1, {\bf B}_2$ is $0$ due to their mutually orthogonal rows and columns. Hence we conclude $R({\bf B}_1 + {\bf B}_2) = R({\bf B}_1) + R({\bf B}_2)$. Since we assume the true parameter has a low rank structure,  we expect the regularized estimator to have a large value of projection on $M({\bf U}, {\bf V})$ and a relatively small valued projection on $N({\bf U}, {\bf V})$.

When the loss function $L(\hat\beta_0, \hat {\bf B}_{\omega_n} )$ defined as $ \frac{1}{2n} \sum_{i=1}^n \Big(y_i - \hat\beta_0 - \langle {\bf X}_i, \hat {\bf B}_{\omega_n} \rangle \Big)^2$ is close to $L(\beta_0, {\bf B}_0)$, it is insufficient to claim $\hat {\bf B}_{\omega_n} - {\bf B}_0$ is small if the loss function $L$ is relatively flat. This is why the strong convexity condition is required.
\begin{definition}
For a given loss function $L$ and norm $\|.\|$, we say $L$ is strong convex with curvature $k_L$ and tolerance function $\tau_L$ if
\[
\delta L({\bf \Delta}, {\bf B}_0) \geq k_L\|{\bf \Delta}\|^2 - \tau_L^2({\bf B}_0), \text{\quad for any }\delta \in \mathcal C(M,N;{\bf B}_0),
\]
where $\mathcal C(M,N;{\bf B}_0) := \{{\bf \Delta} \in \mathcal R^{p\times q} \mid R({\bf \Delta}_{N})\leq 3R({\bf \Delta}_{N^\perp}) + 4R({\bf B}_{0N})\}$.
\end{definition}

Now we are ready to state the main result on the risk bound for our estimate. The proof is provided in the Appendix B.  
\begin{theorem}\label{tm1}
Suppose that (A1) and (A4) hold. Let $\hat{{\bf B}}$ be the solution to \eqref{POLS1}.  If $$\omega_n \geq \frac{12 (\log n)^{3/2} (C_{\mu} + \lambda_u^{1/2}) (\sqrt{p} + \sqrt{q} +  \sqrt{\log n}) }{\sqrt{n}}, $$ then with probability of at least $1 - C n^{-1}$ for some constant $C>0$, 
\begin{align*}
\|\hat {\bf B} -{\bf B}_0\|_F^2 + |\hat \beta - \beta_0^*|^2 \leq 9\frac{\omega_n^2}{\lambda_l}r,
\end{align*} 
where $\beta_0^* = \beta_0 - \pi_2^{-1} \{ c - 1+(\pi_2 - \pi_2^2) ({\bf D}^{\T} {\bf \Sigma}^{-1}{\bf D})\}$ and $c$ is some positive constant. 
\end{theorem}
Theorem \ref{tm1} gives a non-asymptotic risk bound for the proposed estimators. In other words, the results hold for any positive $\omega_n$ satisfying the conditions there. However, in order to ensure the consistency of the proposed estimators, we will need the risk bound to go to $0$, which requires $\omega_n \rightarrow 0$ and $\max(p,q) = o\left(n/(r  \log^3 n)\right)$. If the rank of ${\bf B}_0$ is fixed, then both $p$ and $q$ can diverge with $n$ at the order of $o(n/\log^3 n)$ and their product $p q > n$. This result is compatible with Theorem 1 in \citet{Rask2015}. Note that the estimated intercept $\hat \beta$ converges to $\beta_0^*$, which deviates from the truth $\beta_0$. This is expected because the solution to OLS is only equivalent with LDA's solution in terms of the slope ${\bf B}$, not on $\beta_0$. More precisely, for OLS, by taking the derivative of squared loss function with respect to $\beta_0$ and set it to $0$, we essentially require $E (\epsilon) = 0$. However, this does not hold in our case. Instead we need to shift the residual $\epsilon$ by $d$ to balance off the bias in the cross-product term E$(\epsilon {\bf X})$. The proof of the theorem uses Gaussian comparison inequality which allows us to deal with $\text{vec}({\bf X})$ following a general Gaussian distribution instead of standard Gaussian distribution given that the largest singular value of ${\bf \Sigma}_{xx}$ is bounded. Based on this connection, we further utilize concentration property of spectral norm of Gaussian random matrices. 

Next we show that $\hat{\bf B}$ is rank-consistent under stronger conditions. 

\begin{theorem}\label{thm2}
Suppose that (A1)--(A4) hold. Then the estimate $\hat {\bf B}$ is rank-consistent, that is, $P(\mbox{rank}(\hat {\bf B})=\mbox{rank}({{\bf B}}_0))\rightarrow 1$ as $ n\rightarrow \infty$.
\end{theorem}
Similar to Lasso, estimation consistency does not guarantee correct rank estimation for matrix regularization. In fact, the assumptions here are stronger than those in Theorem \ref{tm1}. For example, Theorem \ref{tm1} allows $p + q = o(n/\log^3n)$ while Theorem \ref{thm2} requires $\max(p,q) = o\left( n^{1/3}  \log^{-3/2} n\right )$ if $\min(p,q) = O(1)$. The proof is based on the arguments in \citet{Bach2008} with modifications to allow diverging $p$ and $q$. 

\begin{remark}
Although nuclear norm penalized least squares is used to estimate the classification direction, there is a fundamental difference between our theorems and the theoretical results derived for nuclear norm penalized least squares regression \citep{Bach2008, Nega2012}. The previous work assumes that the data obey a linear regression model with covariates-independent additive noise, which is not true in our case. In particular, the covariates ${\bf X}_i$ and the residuals $\epsilon_i$ are no longer independent in our problem, which brings additional challenges in developing theoretical results. 
\end{remark}
Next we state a classification error consistency result. To be consistent with the notation in the classification literature, for subject $i$, we use $Y_i \in \{-1,1\}$ to denote its true label, $\hat{f}_n({\bf X_i})$ as the classified label for which $\hat{f}_n$ is the classification rule obtained by solving \eqref{POLS1}, and $l(Y_i ,f({\bf X_i})) = I\{Y_i \neq \text{sign}(f({\bf X_i}))\} $ as the 0-1 loss function. Define the risk of $\hat{f}_n$ as $R(\hat{f}_n) = E_{\bf X} l(Y  ,\hat{f}_n({\bf X}))$ and the Bayes risk as $R^* = \inf_f R(f)$. In addition, we assume that the true label $Y_i$ given ${\bf X_i}$ is determined by the linear classification rule with coefficients $\beta_0^*$ and ${\bf B}_0$. Then the following theorem shows that the proposed classifier achieves the Bayes optimal risk under certain conditions. The proof, given in the Appendix B, is based on the general results in \citet{Zhang2004}, where the author studied the optimal Bayes error rate using  a classifier obtained by minimizing a convex upper bound of the classification error function. 
\begin{theorem}\label{tm3}
Assume the same conditions for Theorem \ref{tm1} hold and $\omega_n \rightarrow 0$. Then $R(\hat{f}_n) \rightarrow R^* $ as $n \rightarrow \infty$.
\end{theorem} 
\section{Numerical results}

\subsection{Simulation}
We conduct simulation studies to evaluate the numerical performance of our proposed method. We compare its performance with that of a few alternatives: ``Lasso LDA'', which adopts a naive Lasso penalty in LDA without taking into account matrix structure, the regularized matrix logistic regression \citep{zhou2014} using nuclear norm and Lasso penalties, denoted by ``Logistic Nuclear'' and ``Logistic Lasso'', and the penalized matrix discriminant analysis (PMDA) approach proposed by \citet{zhong2015matrix}. We generate $n \in \{100,200,500\}$ samples from two classes with weights $(\pi_1,\pi_2) \in \{(0.5,0.5),(0.75,0.25)\}$. For each class, we generate predictors from a bivariate normal distribution with means $\mu_g$, $g = 1, 2$, and covariance ${\bf \Sigma}$. We set $\mu_1 = 0$ and $\mu_2 = {\bf \Sigma} \text{vec}({\bf B}_0)$. The covariance matrix ${\bf \Sigma}$ has a 2D autoregressive structure: $ {\rm cov}({\bf x}_{i_1,j_1}, {\bf x}_{i_2,j_2})=0.5^{|i_1-i_2| + |j_1-j_2|}$ for $ 1\leq i_1\leq p $ and $ 1\leq j_1\leq q $.  The true signal ${\bf B}_0$ is generated based on a 64-by-64 image. We consider three settings: a cross, a triangle and a butterfly. These pictures are shown in Figure \ref{fig2}(a) respectively. In particular, the white color denotes value 0 and black denotes 0$.$05. We apply each fitted model to an independent test data set of size $1000$ and summarize the misclassification rates based on 1000 Monte Carlo replications. The results are contained in Table \ref{tb1}.

\begin{table} 
\centering
\def~{\hphantom{0}}
\caption{Simulation results: misclassification rates ($\%$) and associated standard errors obtained from our method, Lasso LDA, Logistic Nuclear (L-Nuclear), Logistic Lasso (L-Lasso) and penalized matrix discriminant analysis (PMDA) based on 1000 Monte Carlo replications. }{%
\begin{tabular}{lccccccc} \hline
Shape &n&$(\pi_1, \pi_2)$& Ours & Lasso LDA& L-Nuclear &L-Lasso& PMDA\\[5pt]
Cross & 100&(0$.$5,0$.$5)&3$.$65(0$.$02)  &17$.$81(0$.$07) &3$.$70(0$.$02)&19$.$51(0$.$07)&*\\
 & 100&(0$.$75,0$.$25) & 3$.$32(0$.$02) &14$.$89(0$.$05) &6$.$62(0$.$04) &18$.$84(0$.$04)&* \\
 & 200 &(0$.$5,0$.$5)&3$.$22(0$.$02) &11$.$69(0$.$05) &3$.$26(0$.$02) &13$.$39(0$.$05)&26$.$93(0$.$05) \\
 & 200 &(0$.$75,0$.$25)&2$.$87(0$.$02) &9$.$89(0$.$04) &4$.$14(0$.$03) &16$.$27(0$.$04)&19$.$58(0$.$08) \\
 & 500 &(0$.$5,0$.$5)&3$.$09(0$.$02) & 6$.$97(0$.$03) &3$.$11(0$.$02) &8$.$19(0$.$04)&25$.$17(0$.$04) \\
 & 500 &(0$.$75,0$.$25)&2$.$62(0$.$02) &5$.$81(0$.$03) &3$.$59(0$.$02) &14$.$91(0$.$03)&12$.$05(0$.$04) \\[ 10pt]
Triangle & 100&(0$.$5,0$.$5)  &3$.$12(0$.$02) &15$.$73(0$.$06) &3$.$11(0$.$02) &17$.$70(0$.$07)&* \\
 & 100&(0$.$75,0$.$25) & 2$.$66(0$.$02)&13$.$72(0$.$05) &6$.$10(0$.$04) &17$.$19(0$.$04)&* \\
 & 200 &(0$.$5,0$.$5)&2$.$85(0$.$02)  &9$.$90(0$.$04) &2$.$81(0$.$02) &11$.$81(0$.$04)& 30$.$17(0$.$08)\\
 & 200 &(0$.$75,0$.$25)&2$.$43(0$.$02)  &8$.$72(0$.$03) &3$.$62(0$.$02) &13$.$40(0$.$04)&24$.$63(0$.$10) \\
 & 500 &(0$.$5,0$.$5)&2$.$67(0$.$02)  & 5$.$67(0$.$03)&2$.$73(0$.$02) &6$.$96(0$.$03)&25$.$92(0$.$04)\\
& 500 &(0$.$75,0$.$25)&2$.$29(0$.$01) &4$.$89(0$.$02) &2$.$74(0$.$02) &9$.$97(0$.$03) &14$.$69(0$.$05)\\ [ 10pt]
Butterfly & 100&(0$.$5,0$.$5)  &3$.$86(0$.$02) &17$.$10(0$.$06) &4$.$16(0$.$02) &18$.$82(0$.$07)&* \\
 & 100&(0$.$75,0$.$25) & 3$.$47(0$.$02)&14$.$79(0$.$05) &7$.$14(0$.$04) &17$.$78(0$.$04)&* \\
 & 200 &(0$.$5,0$.$5)&3$.$67(0$.$02)  &11$.$00(0$.$04) &3$.$78(0$.$02)&12$.$66(0$.$05)&29$.$79(0$.$07)\\
 & 200 &(0$.$75,0$.$25)&3$.$26(0$.$02) &9$.$80(0$.$04) &4$.$50(0$.$02) &13$.$93(0$.$04)&23$.$83(0$.$09) \\
 & 500 &(0$.$5,0$.$5)&3$.$56(0$.$02)  & 6$.$50(0$.$03) &3$.$52(0$.$02) &7$.$70(0$.$03)&25$.$77(0$.$04) \\
 & 500&(0$.$75,0$.$25) & 3$.$02(0$.$02)&5$.$74(0$.$03) &3$.$51(0$.$02) &10$.$49(0$.$03)&14$.$66(0$.$05) \\ \hline
\end{tabular}}
\label{tb1}
\end{table} 
The results show that our method performs much better than ``Lasso LDA" and ``Logistic Lasso" under all scenarios. This is expected because these two methods ignore the matrix structure. For ``Logistic Nuclear'', it has similar misclassification rates with our method for balanced data, but does not perform as good as ours for unbalanced data. We have also plotted the estimates using nuclear norm and $\ell_1$-norm from one randomly selected Monte Carlo replicate in Figure \ref{fig2}(b)(c). It can be seen that the proposed nuclear norm regularization is much better than $\ell_1$-regularization in  recovering the matrix signal in different shapes. By comparing the recovery of different shapes in Column (b) in Figure \ref{fig2}, we find that our method works better for cross than for triangle and butterfly. This is expected since triangle and butterfly do not have the low rank structure.  

We also compare the performance of our method with that of PMDA proposed by \citet{zhong2015matrix}. In Table \ref{tb1}, it can be seen that our proposed method has a lower mis-classification rate under all scenarios. This is because we allow flexible values of the rank for the linear discriminant direction $ {\bf B} $, while in \citet{zhong2015matrix}, their assumption is equivalent to assuming $ {\bf B} $ is of rank 1. In particular, using their notation, for binary case, their direction $ {\bf B}=\boldsymbol\beta_1\boldsymbol\xi^{\T}$, where $\boldsymbol\beta_1 \in \mathbb{R}^{p}$ and $\boldsymbol\xi\in \mathbb{R}^{q}$. Since the true ranks of $ {\bf B} $ in our simulation studies are all of rank greater than 1, it is not surprising that our method outperforms PMDA. Moreover, PMDA does not apply to the case where $n < p + q$, i.e., the sample size is far smaller than the summation of image dimensions. Therefore, their method does not apply to one of our simulation settings $(n,p,q)=(100,64,64)$ and we mark their results using $*$ in Table \ref{tb1}. We also compare the computation time between PMDA and our method. In simulation, when $n=200$ and true signal is a cross, given a fixed regularization parameter, the system running time of PMDA is around 1.5 minutes whereas the system running time of our method is no more than 13 seconds. Here system running time is measured on a Macbook Pro laptop with a 2.9 GHz Intel Core i5. This is because PMDA essentially solves least square problems with $L_1$ penalty in each iteration when setting $\omega_1 =0$ in Algorithm 2 in \citet{zhong2015matrix}. Our method is based on the Nesterov optimal gradient method which avoids computing inverse of covariance matrix and hence has a faster convergence rate.

\begin{figure}[h]
\begin{tabular}{ccc}
\includegraphics[width=0.3\linewidth]{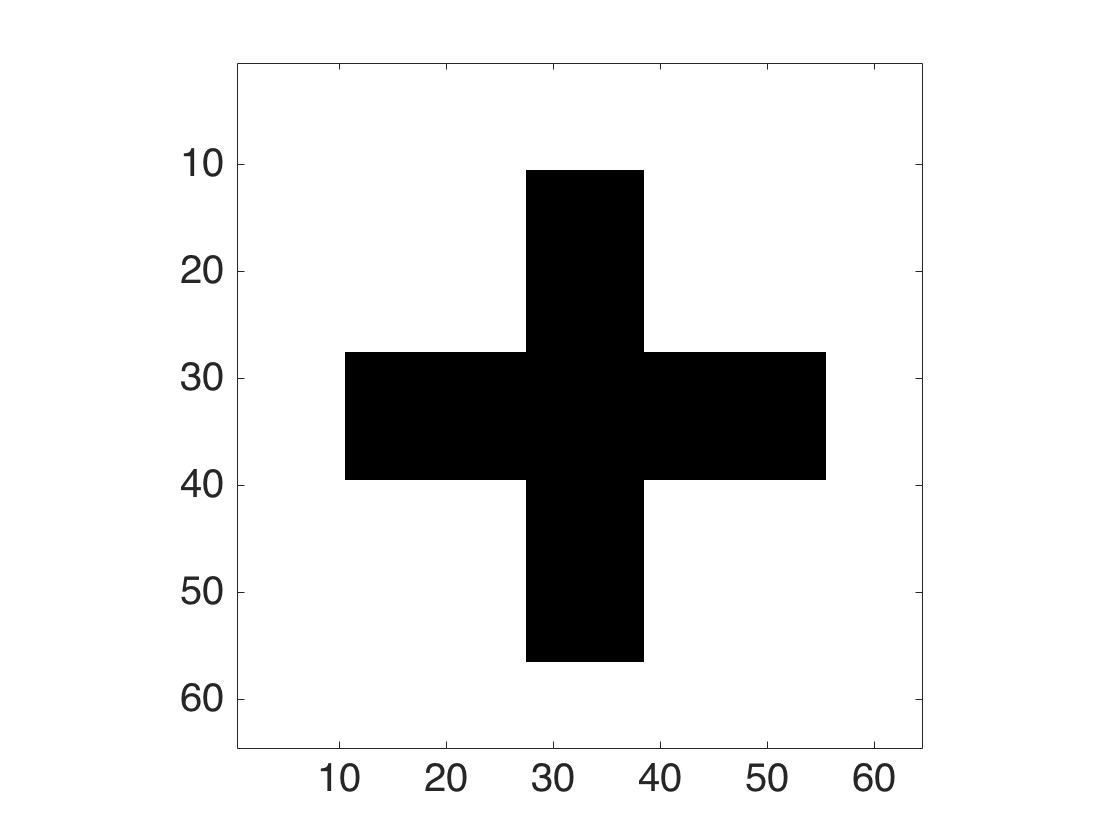}\label{(a)}&
\includegraphics[width=0.3\linewidth]{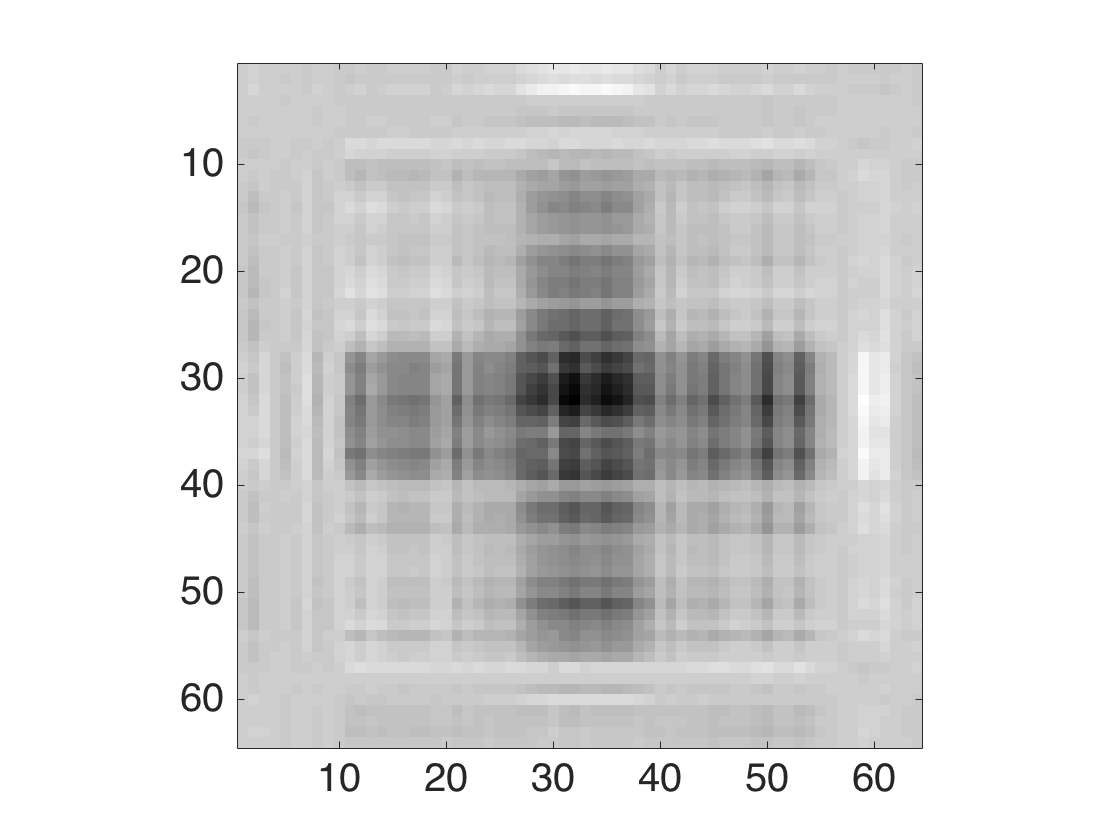}\label{(b)}&
\includegraphics[width=0.3\linewidth]{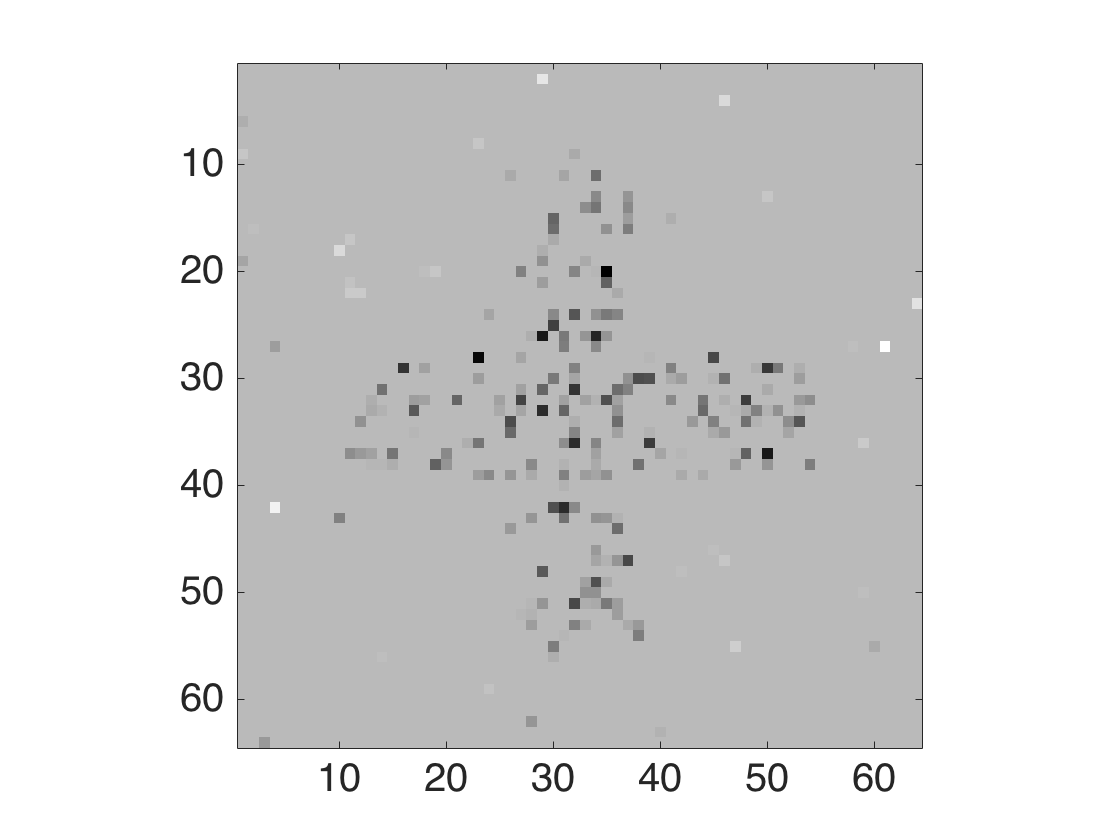}\label{(c)}\\
\includegraphics[width=0.3\linewidth]{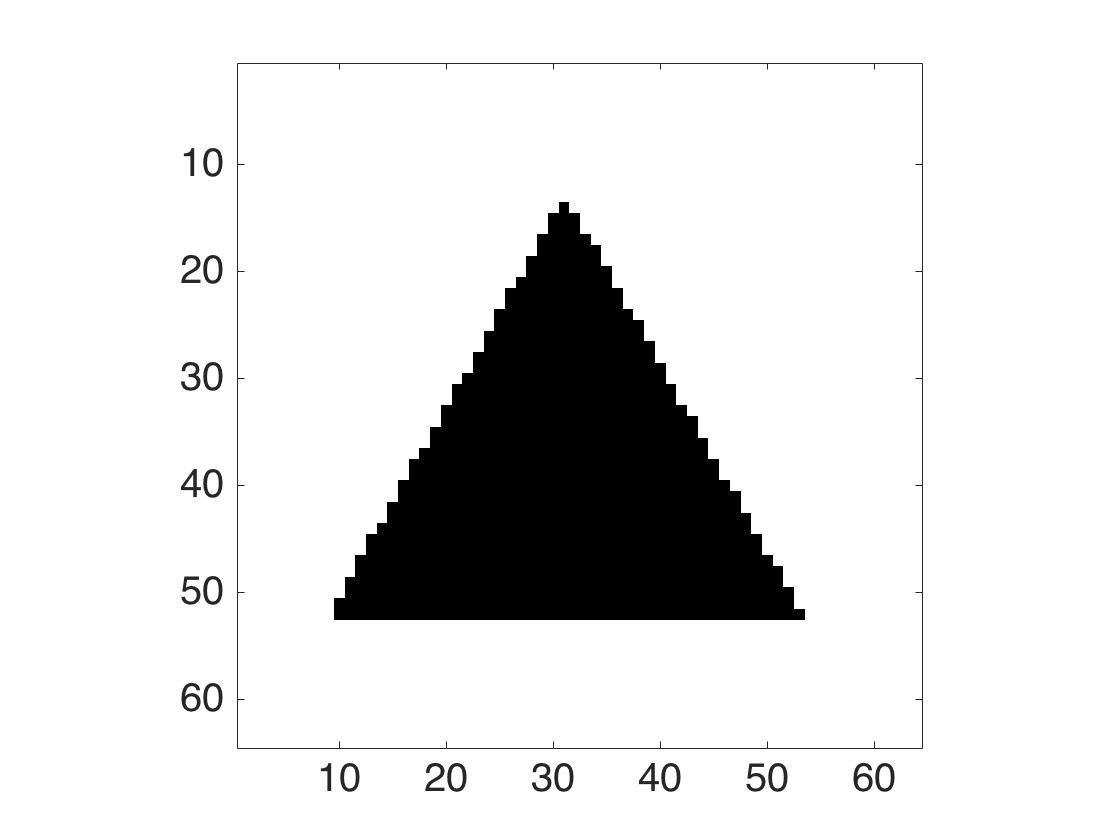}\label{(a)}&
\includegraphics[width=0.3\linewidth]{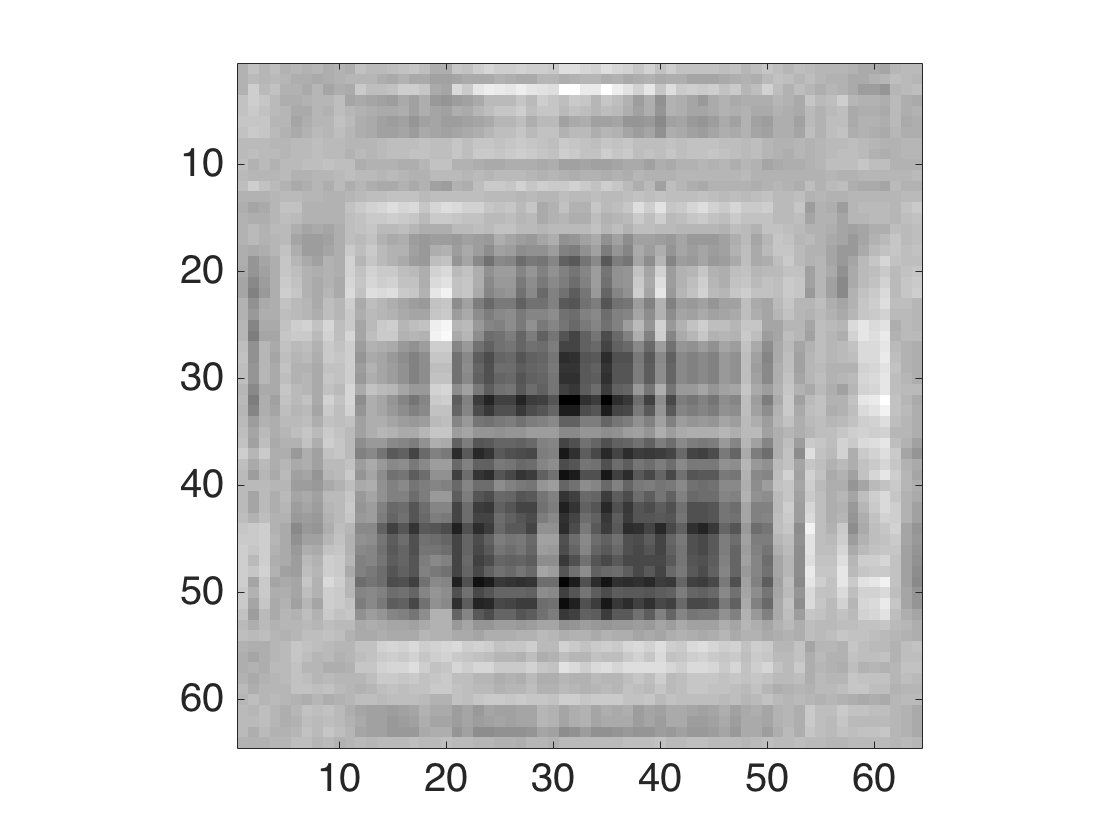}\label{(b)}&
\includegraphics[width=0.3\linewidth]{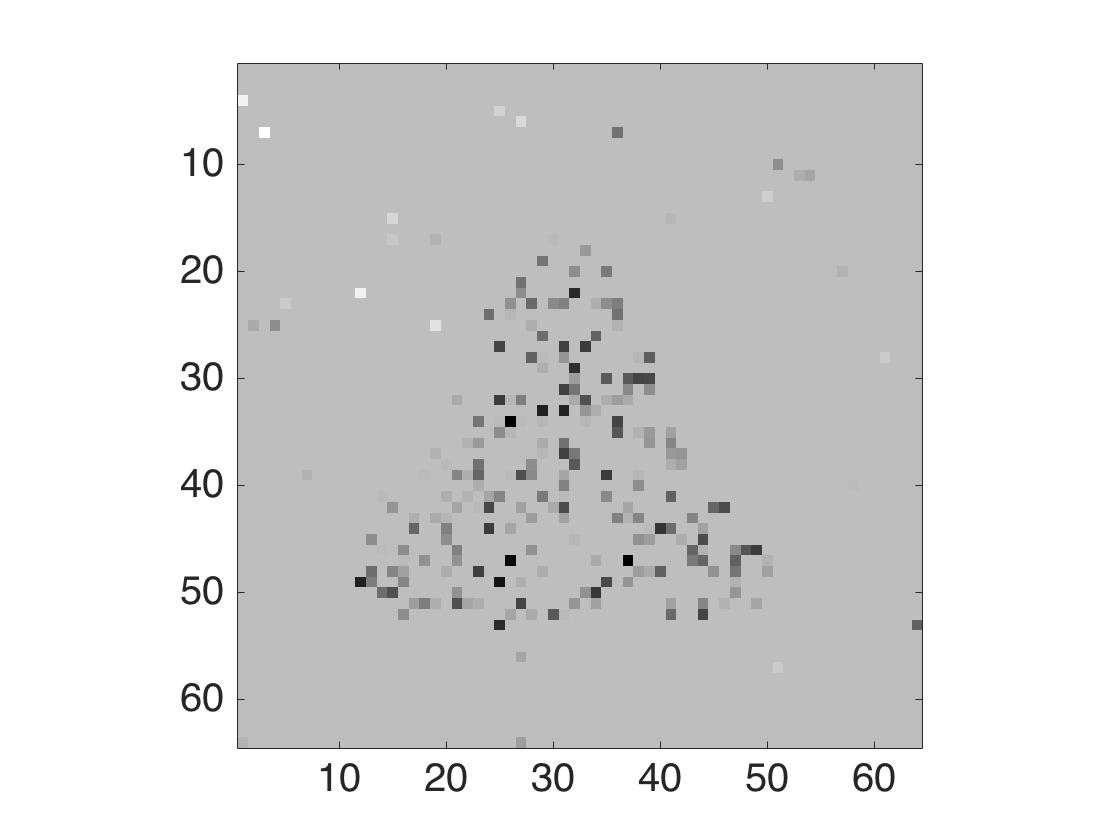}\label{(c)}\\
\includegraphics[width=0.3\linewidth]{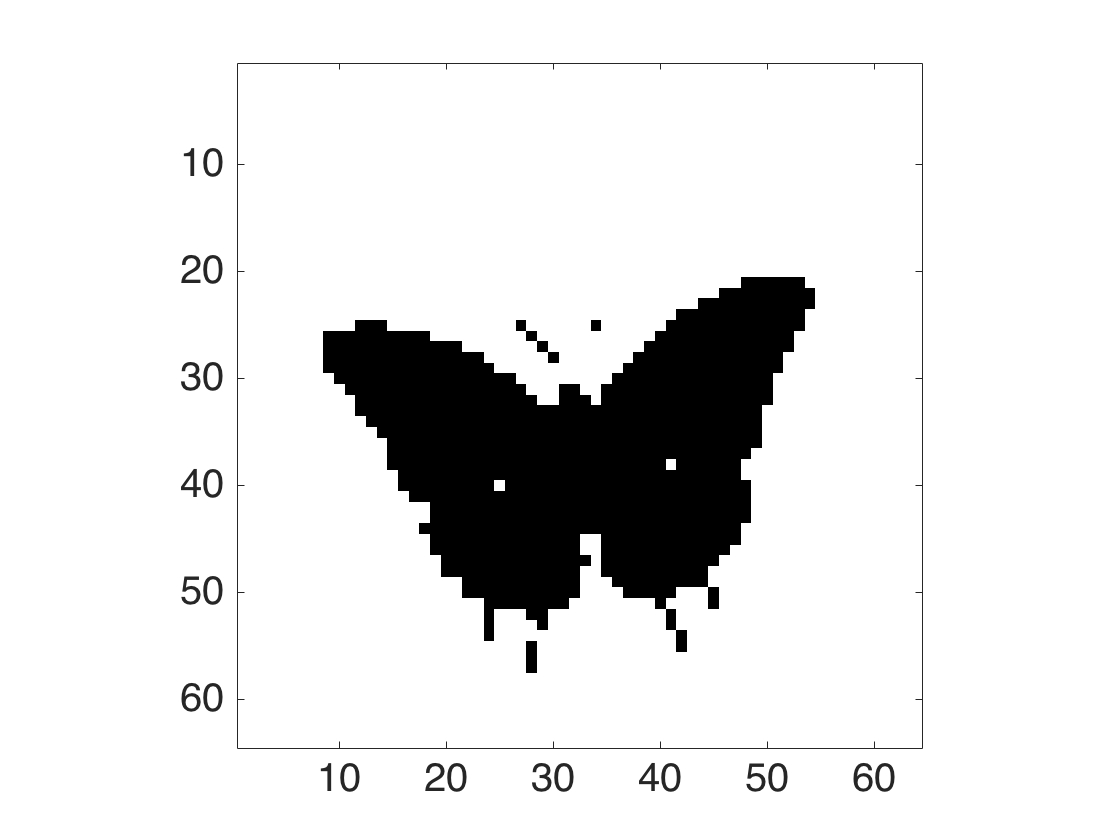}\label{(a)}&
\includegraphics[width=0.3\linewidth]{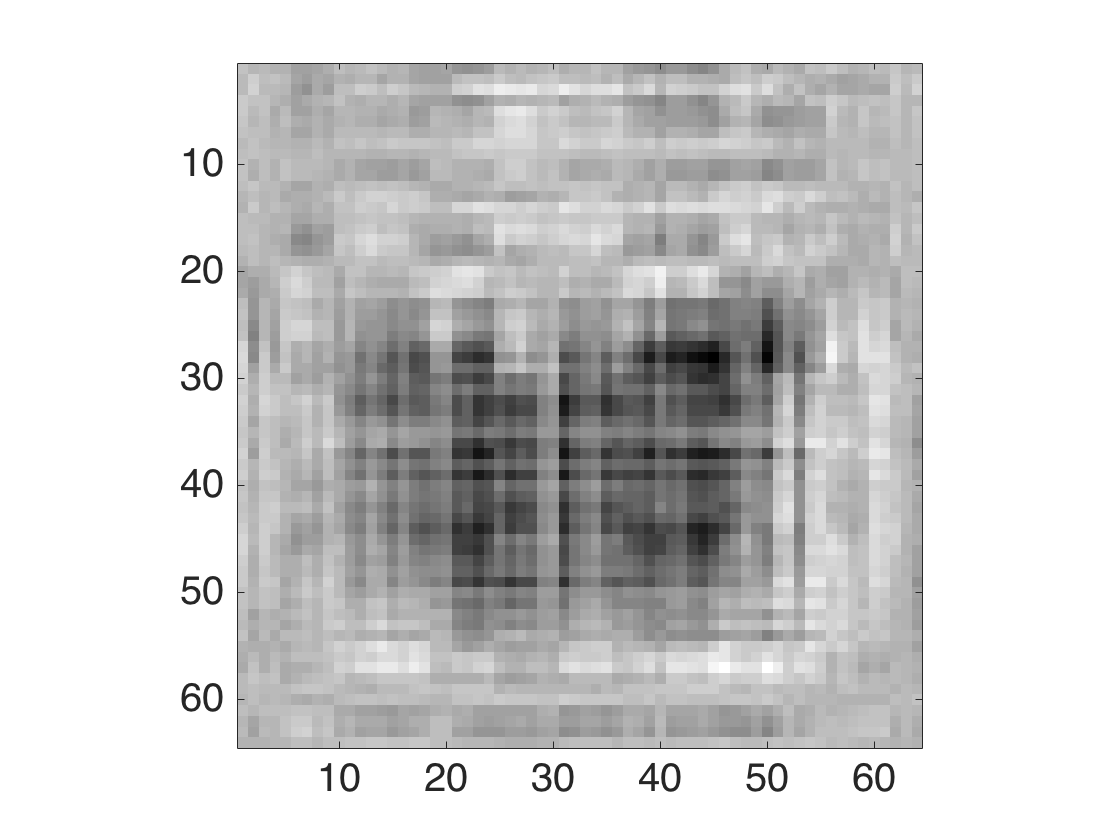}\label{(b)}&
\includegraphics[width=0.3\linewidth]{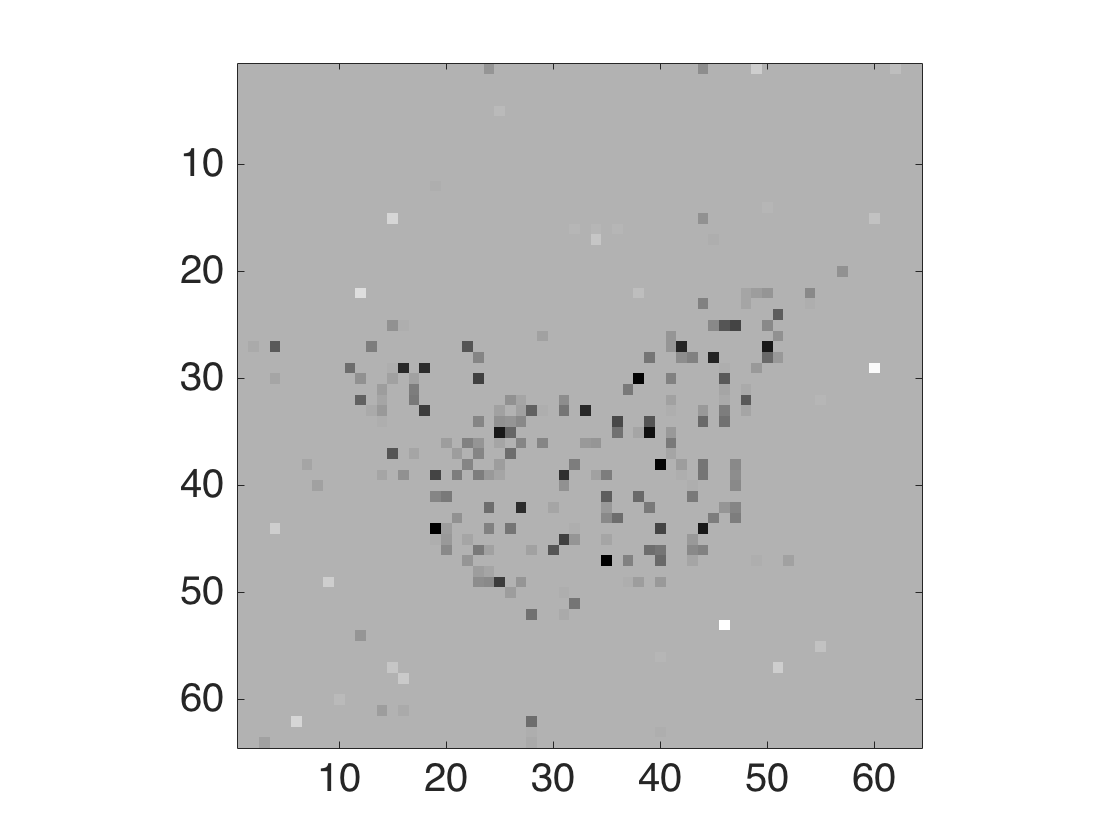}\label{(c)}\\
(a)&(b)&(c)\\
\end{tabular}
\caption{The figures for cross image: (a) original signal; (b) our nuclear regularization estimate; (c) $\ell_1$-regularized estimate.}\label{fig2}
\end{figure}
%\begin{figure}
%\subfigure[]{\includegraphics[width=0.3\linewidth]{f1}}
%\subfigure[]{\includegraphics[width=0.3\linewidth]{f2}}
%\subfigure[]{\includegraphics[width=0.3\linewidth]{f3}}
%\caption{The figures for cross image: (a) original signal; (b) our nuclear regularization estimate; (c) $\ell_1$-regularized estimate.}\label{fig2}
%\end{figure}

\subsection{Real data application}
We apply our method to an EEG dataset, which is available at  \url{https://archive.ics.uci.edu/ml/datasets/EEG+Database}. The data was collected by the Neurodynamics Laboratory to study the EEG correlates of genetic predisposition to alcoholism. It contained measurements from 64 electrodes placed on each subject's scalps sampled at 256 Hz (3$.$9-msec epoch) for 1 second.  Each subject was exposed to three stimuli: a single stimulus, two matched stimuli, two unmatched stimuli. Among the 122 subjects in the study, 77 were alcoholic individuals and 45 were controls. More details about the study can be found in \citet{zhang1995}. In statistics literature, EEG data has been analyzed using different models, for example, \citet{Gao20191} considered an unsupervised approach for clustering EEG data, \citet{Gao20192} and \citet{Gao2018} considered an evolutionary state-space model and graphical model for better understanding brain connectivity, respectively. However, these methods are not directly applicable for classification purpose here. 

In our data analysis, for each subject, we use the average of all 120 runs for each subject under single-stimulus condition and use that as the covariate ${\bf x}_i$, which is a $256\times64$ matrix. The classification label is \emph{alcoholic} or not. 
We randomly divide the data set into training set of 81 subjects and test set of 41 subjects for 100 times, and each time fit the model on the training set and apply it on the test set to obtain the average mis-classification rate and its standard error. The results for different methods are summarized in Table \ref{tb2}. It can be seen that the proposed method has a significant lower mis-classification rate compared with other methods, which agrees with the simulation findings for the unbalanced data. PMDA does not work here since $p+q >n $ ($(n,p,q)=(122, 256, 64)$). We also check the fitted signal matrix and it agrees well with the one obtained by \citet{zhou2014}. 

In terms of computational efficiency,  we measured the computation time among Lasso LDA, Logistic Nuclear, Logistic Lasso and our method based on one evaluation of the data, that is, partitioning the data into training and test sets, fitting the model on the training set and applying it on the test set. The running time for Lasso LDA, Logistic Nuclear, Logistic Lasso and our method is 0.67s, 1.79s, 1.27s and 1.87s, respectively. The system running time is measured in Matlab R2015b on a Macbook Pro  laptop with a 2.9 GHz Intel Core i5.

\begin{table}
\centering
\def~{\hphantom{0}}
\caption{EEG data analysis: misclassification rates ($\%$) and associated standard errors.}{%
\begin{tabular}{lrrrr}
Our method & Lasso LDA & Logistic Nuclear& Logistic Lasso & PMDA\\
22$.$20(0$.$53)&24$.$12(0$.$70)&24$.$44(0$.$80)&26$.$24(0$.$91)&*
\end{tabular}}
\label{tb2}
\end{table}

\begin{comment}
 \begin{table} 
 \centering
\def~{\hphantom{0}}
\caption{Electroencephalography data analysis: misclassification rates ($\%$).}{%
\begin{tabular}{lcccc}
Method& leave-one-out &5-fold &10-fold &20-fold \\
Our method&0$.$205&0$.$172&0$.$189&0$.$180\\
Lasso LDA& 0$.$238& 0$.$221&0$.$250 &0$.$262\\
Logistic Nuclear& 0$.$230&0$.$214&0$.$222&0$.$181\\
Logistic Lasso&0$.$246&0$.$287&0$.$271&0$.$264\\
\end{tabular}}
\label{tb2}
\end{table}
\end{comment}
\section{Discussion}
In the literature, total variation (TV) regularization has also been commonly used for modeling image data in addition to the proposed nuclear norm regularization. Their focuses are slightly different -- the former is on structured sparse pattern and the later is on low-rank pattern. The main reason that we choose to focus on the nuclear norm regularization in this paper is because we have found that low rankness is a more reasonable assumption than sparseness assumption in our real data application. In particular, the mis-classification errors of our method is lower than the sparse method (LASSO) in our real data analysis. The TV regularization is an interesting direction to explore as it requires new computational algorithms and theories; and thus we leave this for the future research.  

In this paper, we only consider the situation where all the image measurements are taking at the same scale, that is, the dimension of the image covariates $p$ and $q$ are equal for every study subject. We believe this is the case for most applications. For the special cases when image dimensions vary across subjects, our method may still be applicable by first resizing the image to the same scale. It will be of future interest to develop flexible statistical methods to handle image data that can be of different sizes in general.

\section*{Acknowledgment}
The authors would like to thank the Editor, Associate Editor and two reviewers for their constructive comments, which have substantially improved the paper. Shen's research is partially supported by Simons Foundation Award 512620 and NSF DMS-1509023. Zhou's research is partially supported by NIH grants R01HG006139, R01GM53275 and NSF DMS-1310319. Kong's research
is partially supported by the Natural Science and Engineering Research Council of Canada.

%\bigskip
%\begin{center}
%{\large\bf SUPPLEMENTARY MATERIAL}
%\end{center}
\appendix
\section{Primary lemmas and propositions}\label{app:theorem}

% Note: in this sample, the section number is hard-coded in. Following
% proper LaTeX conventions, it should properly be coded as a reference:

%In this appendix we prove the following theorem from
%Section~\ref{sec:textree-generalization}:
 
We start with some useful lemmas in this section. The proof of main theorems are given in the Appendix B. 

\noindent
 
We first re-state a singular value thresholding formula in \citet{cai2010singular}. This result is extremely useful when computing optimal solution of  \eqref{POLS}, by which the important block of Nestorov's algorithm was formed. The proof is based on showing that $0$ is one of subgradients of \eqref{svdobject} at $\hat {\bf B}$. 
\begin{proposition}\label{prop:svdthredsholding}
For any $\omega \geq 0$ and a given matrix ${\bf B}_0 \in \mathcal R^{p\times q}$ with singular value decomposition $U\text{diag}(s) V^{\T}$, the minimizer $\hat {\bf B}$ of
\begin{equation}
\frac{1}{2}\|{\bf B} - {\bf B}_0\|_F^2 + \omega \|{\bf B}\|_\ast
\label{svdobject}
\end{equation}
has the same singular vectors as $B_0$ with singular values $(s_i - \omega)_{+}$.
\end{proposition}
\noindent

Next we state a lemma on the risk bound. This result can be viewed as an analog of Theorem 1 in \citet{Nega2012} under our situation. 
\begin{lemma}\label{lemma:nega}
Suppose that (A1) and (A2) hold, and $\omega_n \geq 2 \|\frac{1}{n} \sum_{i=1}^n \epsilon_i {\bf X}_i \|_2$. Then any optimal solution $\hat{{\bf B}} $ to 
\begin{align} \label{POLS}
(\hat\beta_0, \hat {\bf B}) = \argmin_{\beta_0, {\bf B}}\frac{1}{2n} \sum_{i=1}^n \Big(y_i - \beta_0 - \langle {\bf X}_i, {\bf B} \rangle \Big)^2 + \omega_n\|{\bf B}\|_\ast
\end{align}
satisfies the bound
\begin{align*}
\|\hat{{\bf B}}  - {\bf B}_0\|_F^2 \leq 9\frac{\omega_n^2}{\lambda_l} r.
\end{align*}
\end{lemma}
\begin{proof} 
We apply Theorem 1 in \citet{Nega2012} to our situation. Observe that the nuclear norm is decomposable, and the squared error loss satisfies $\tau_{\mathcal{L}}({\bf B}_0) = 0$ in that paper. Moreover, the dual norm $\mathcal{R}^*$ to the nuclear norm is simply the spectral norm. The curvature constant $\kappa_{\mathcal{L}}$ in the restricted strong convexity (RSC) condition can be chosen as $\lambda_l^{1/2}$ because the squared error loss is used and the Hessian matrix $\text E\{\text{vec}({\bf X})\text{vec}({\bf X})^{\T}\} = {\bf \Sigma}_{xx} \geq \lambda_l I$.  For a subspace $M$ that contains matrices of the rank at most $r$, its subspace compatibility constant satisfies
\[
\psi(M) = \sup_{{\bf U}\in M\setminus \{0\}} \frac{\|{\bf U}\|_\ast}{\|{\bf U}\|_F} = \sup_{{\bf U}\in M\setminus \{0\}}\frac{\sum_{i=1}^r\sigma_i({\bf U})}{(\sum_{i=1}^r \sigma_i({\bf U})^2)^{1/2}} \leq \sqrt r,
\]
where the last inequality follows by Cauchy-Schwarz inequality. Hence subspace compatibility constant under the low-rank assumption (A2) is bounded by $\sqrt r$. 
\end{proof}

Next we state a few commonly used lemmas regarding the concentration property and tail probability inequalities of Gaussian and sub-Gaussian random variable (matrices).  Their proofs can be found in standard textbooks, e.g., \citet{wainwright2015high}. 
\begin{lemma}\label{lemma:Hoeffding}
(Hoeffding bound) Suppose that the variables ${\bf X}_i$, $i = 1, 2, \ldots, n$ are independent and $X_i$ has mean $\mu_i$ and sub-Gaussian parameter ${\bf \Sigma}_i$. Then for all $t\geq 0$, we have
\[
P\left(\sum_{i=1}^n({\bf X}_i - \mu_i) \geq t\right) \leq \exp(-\frac{t^2}{2\sum_{i=1}^n{\bf \Sigma}_i^2})
\]
\end{lemma}
\begin{lemma}\label{lemma:matrix-con}
Assume ${\bf X}_1,\ldots,{\bf X}_n \in \mathbb{R}^{p \times q}$ are i.i.d. random matrices. Suppose that $\|{\bf X}_1 \|_2 \leq M$ almost surely, then with probability greater than $1 - \delta$, 
\begin{align*}
\left\|\frac{1}{n} \sum_{i=1}^n {\bf X}_i    - E {\bf X}_1  \right\|_2 \leq \frac{6M}{\sqrt{n}} \left( \sqrt{\log \min(p,q)} + \sqrt{\log (1/\delta)}\right)
\end{align*}
\end{lemma}
\begin{lemma}\label{lemma:gordon} Let ${\bf A}$ be an $p\times q$ matrix whose entries are independent standard normal random variables. Denote $s_{\min}({\bf A})$ and $s_{\max}({\bf A})$   as smallest singular value and largest singular value of $\bf A$ respectively.  Assume $p\geq q$ without loss of generality. Then
\[
\sqrt p  - \sqrt q  \leq E_{s_\text{min}}({\bf A}) \leq  E_{s_\text{max}}({\bf A})\leq \sqrt p +\sqrt q.
\]
\end{lemma}

\begin{lemma}\label{lemma:concentration}
Let ${\bf Y}\sim N(0,I_{d\times d})$ be a d-dimensional Gaussian random variable. Then for any function F: $\mathcal{R}^d \to \mathcal R$ with Lipschitz constant L, i.e. $|F({\bf x}) - F({\bf y})| \leq L\|{\bf x} - {\bf y}\|$ for all ${\bf x}, {\bf y} \in \mathcal R^d$, we have
\[
P\left\{|F({\bf Y}) - E(F({\bf Y}))| \geq t\right\} \leq 2\exp(-\frac{t^2}{2L^2}),
\] for any $t>0$.
\end{lemma}
 
\begin{lemma}\label{lemma:comparison}
(Anderson's comparison inequality \citep{anderson1955integral}) Let ${\bf X}$ and ${\bf Y}$ be zero-mean Gaussian random vectors with covariance ${\bf \Sigma}_{\bf X}$ and ${\bf \Sigma}_{\bf Y}$ respectively. If ${\bf \Sigma}_{\bf X} - {\bf \Sigma}_{\bf Y}$ is positive semi-definite then for any convex symmetric set C, 
\[
P({\bf X} \in C) \leq P({\bf Y} \in C).
\]
\end{lemma}

The following lemma is very useful in establishing rank estimation consistency. 
\begin{lemma}\label{prop1}
Assume (A1) and (A2) hold.  Let $\hat {\bf B}$ be a global minimizer of \eqref{POLS}. If $n^{1/2}\omega_n$ tends to $+\infty$ and $\omega_n$ tends to zero, then $\omega_n^{-1}(\hat {\bf B} - {\bf B}_0)$ converges in probability to the unique global minimizer ${\bf \Delta}$ of 
\[
\min_{{\bf \Delta} \in \mathrm{R}^{p\times q} }\frac{1}{2}\text{vec}({\bf \Delta})^{\T}{\bf \Sigma} \text{vec}({\bf \Delta}) +
\text{tr} \{{\bf U}_0^{\T}{\bf \Delta} {\bf V}_0\} + \|{\bf U}_{0\perp}^{\T} {\bf \Delta} {\bf V}_{0\perp}\|_\ast.
\]
Moreover, $\hat {\bf B} = {\bf B}_0 + \omega_n{\bf \Delta} + O_p\big (\omega_n\min(p,q)n^{-1/2} + \min(p,q) n^{-1/2} +\omega_n^{2}\min(p,q)^{1/2}n^{-1/2}\big )$.
\end{lemma}

\begin{proof} 
We can write $\hat {\bf B} = {\bf B}_0 + \omega_n\hat{\bf \Delta}, $ where $\hat {\bf \Delta}$ is the global minimum of
\[
%\begin{split}
V_n({\bf \Delta}) = \frac{1}{2}\text{vec}({\bf \Delta})^{\T}\hat{\bf \Sigma}_{xx}\text{vec}({\bf \Delta}) - \omega_n^{-1}\text{tr}{\bf \Delta}^{\T}
\hat{\bf \Sigma}_{{\bf X}\epsilon} + \omega_n^{-1}(\|{\bf B}_0 + \omega_n{\bf \Delta}\|_\ast - \|{\bf B}_0\|_\ast),
%\end{split}
\] where $\hat{\bf \Sigma}_{xx} = n^{-1} \sum_{i=1}^n \text{vec}({\bf X}_i) \text{vec}({\bf X}_i)^{\T}$ and $\hat{\bf \Sigma}_{{\bf X}\epsilon} = n^{-1} \sum_{i=1}^n \epsilon_i \text{vec}({\bf X}_i)$. Then $\text{vec}({\bf \Delta})^{\T}\hat{\bf \Sigma}_{xx}\text{vec}({\bf \Delta})/2 -  \text{vec}({\bf \Delta})^{\T}{\bf \Sigma}_{xx} \text{vec}({\bf \Delta})/2$ converges to $\text{vec}({\bf \Delta})^{\T} \text{E}(\hat{\bf \Sigma}_{xx} - {\bf \Sigma}_{xx} )\text{vec}({\bf \Delta})/2$ with probability of 1. Note that $\text{E}\|\hat{\bf \Sigma}_{xx} - {\bf \Sigma} \|_F^2  = O(n^{-1})$. Denote $\text{vec}({\bf \Delta})_i$ as $a_i$ and $(\hat{\bf \Sigma}_{xx} - {\bf \Sigma})_{ij}$ as $b_{ij}$. Then we have
\[
\begin{split}
\frac{1}{2}|\text{vec}({\bf \Delta})^{\T} \text{E}(\hat{\bf \Sigma}_{xx} - {\bf \Sigma} )\text{vec}({\bf \Delta})|& \leq \sum_{i,j =1}^{pq} |a_i a_j \text{E}(b_{ji})|\\
&\leq \Big(\sum_{i,j =1}^{pq} a_i^2 a_j^2 \sum_{i,j =1}^{pq} \text{E}(b_{ij}^2)\Big)^{\frac{1}{2}}\\
&=\sum_{i=1}^{pq} a_i^2 \text{E}(\sum_{i,j=1}^{pq} b_{ij}^2)^{\frac{1}{2}}\\
&= \sum_{i=1}^{pq} a_i^2 \text{E}\|\hat {\bf \Sigma}_{xx} - {\bf \Sigma}_{xx}\|_F\\
&=\|{\bf \Delta}\|_F^2 O(n^{-1/2})\\
& = O\big (\min(p, q) \|{\bf \Delta}\|_2^2 n^{-1/2}\big ).
\end{split}
\]
Meanwhile, 
\[
\begin{split}
|\text{tr}{\bf \Delta}^{\T}\hat{\bf \Sigma}_{{\bf X}\epsilon}| & \leq (\text{tr}{\bf \Delta}^{\T}{\bf \Delta})^{\frac{1}{2}}(\text{tr}\hat{\bf \Sigma}_{{\bf X}\epsilon}^{\T}\hat {\bf \Sigma}_{{\bf X}\epsilon})^{\frac{1}{2}}\\
&= \|{\bf \Delta}\|_F O_p(n^{-1/2})\\
&\leq \min(p,q)^{\frac{1}{2}} \|{\bf \Delta}\|_2 O_p(n^{-\frac{1}{2}}).
\end{split}
\]
Therefore 
\[
\begin{split}
V_n({\bf \Delta}) = &\frac{1}{2}\text{vec}({\bf \Delta})^{\T}{\bf \Sigma} \text{vec}({\bf \Delta}) + O_p\big (\min(p,q)n^{-1/2} \|{\bf \Delta}\|_2^2\big )
+ O_p\big (\min(p,q)^\frac{1}{2}\omega_n^{-1}n^{-1/2}\|{\bf \Delta}\|_2\big )\\
& + \text{tr}({\bf U}_0^{\T}{\bf \Delta} {\bf V}_0) + \|{\bf U}_{0\perp}^{\T}{\bf \Delta} {\bf V}_{0\perp}\|_\ast + O_p\big (\omega_np^{1/2}q^{1/2}\min(p,q)\|{\bf \Delta}\|_2^2\big ).\\
=&V({\bf \Delta}) + O_p\big (\min(p,q)n^{-1/2}\|{\bf \Delta}\|_2^2\big ) + O_p\big(\min(p,q)^\frac{1}{2}\omega_n^{-1}n^{-1/2}\|{\bf \Delta}\|_2\big) \\
& ~~~~ + O_p\big(\omega_np^{1/2}q^{1/2}\min(p,q)\|{\bf \Delta}\|_2^2\big),
\end{split}
\]
where $p^{1/2}q^{1/2}$ in the last term comes from the Frobenius norm of any matrix in $\mathcal R^{p\times q}$ with bounded entries. 
Let $s_r$ be the $r$-th largest singular value of $B_0$, for any $M < s_r/(2\omega_n)$,
\[
\begin{split}
& \text{E}~ \text{sup}_{\|{\bf \Delta}\|_2\leq M} | V_n({\bf \Delta}) - V({\bf \Delta}) | \\
& ~~~= O \big( \min(p,q)M^2 \text{E}
\|\hat{\bf \Sigma}_{xx} - {\bf \Sigma} \|_F + M\min(p,q)^\frac{1}{2}\omega_n^{-1}E(\|\hat{\bf \Sigma}_{M\epsilon}\|^2)^{1/2}  + \omega_np^{1/2}q^{1/2} \min(p,q)M^2\big)\\
& ~~~ = fO\big( \min(p,q)M^2 n^{-1/2} + M\min(p,q)^\frac{1}{2}\omega_n^{-1}n^{-1/2} + \omega_np^{1/2}q^{1/2}\min(p,q) M^2\big).
\end{split}
\]
Obviously $V({\bf \Delta})$ achieves its minimum in the bounded ball at ${\bf \Delta}_0 \neq 0$. Hence by Markov inequality the probability of the minimum of $V_n({\bf \Delta})$ lying strictly inside the ball $\|{\bf \Delta}\|_2 < 2\|{\bf \Delta}_0\|_2$ tends to one and is also the unconstrained minimum.
\end{proof}
 
 The following two lemmas can be viewed as analogs of Proposition 3 and Lemma 11 in \citet{Bach2008}. W present them without the proof. 
 
\begin{lemma}\label{lem:sic}
Let ${\bf B}_0 = {\bf U}_0\text{Diag}(S_0){\bf V}_0^{\T}$ be the singular value decomposition of ${\bf B}_0$. Then the unique global minimizer of
\begin{equation*}\label{eq4}
\frac{1}{2}\text{vec}({\bf \Delta})^{\T}{\bf \Sigma}\text{vec}({\bf \Delta}) +\text{tr}{\bf U}_0^{\T}{\bf \Delta}
{\bf V}_0 +\|{\bf U}_{0\perp}^{\T}{\bf \Delta} {\bf V}_{0\perp}\|_\ast
\end{equation*}
satisfies ${\bf U}_{0\perp}^{\T}{\bf \Delta} {\bf V}_{0\perp} = 0$ if and only if
\[
\Bigg \|\{({\bf V}_{0\perp}\otimes {\bf U}_{0\perp})^{\T}{\bf \Sigma}^{-1}({\bf V}_{0\perp}\otimes {\bf U}_{0\perp})\}^{-1}
\{({\bf V}_{0\perp}\otimes {\bf U}_{0\perp})^{\T}{\bf \Sigma}^{-1}({\bf V}_0\otimes {\bf U}_0)\text{vec}({\bf I})\}\Bigg \|_2 \leq 1.
\]
Furthermore, when ${\bf U}_{0\perp}^{\T}{\bf \Delta} {\bf V}_{0\perp} = 0$, the solutions has these forms:
\[
\text{vec}(\Lambda) = \{({\bf V}_{0\perp}\otimes {\bf U}_{0\perp})^{\T}{\bf \Sigma}  ({\bf V}_{0\perp}\otimes {\bf U}_{0\perp})\}^{-1}\{({\bf V}_{0\perp}\otimes {\bf U}_{0\perp})^{\T}{\bf \Sigma}  ({\bf V}_0\otimes {\bf U}_0)\text{vec}(I)\},
\]
\begin{equation}\label{deltasolution}
\text{vec}({\bf \Delta}) =  -{\bf \Sigma}^{-1} \text{vec}({\bf U}_0{\bf V}_0^{\T} - {\bf U}_{0\perp}\Lambda {\bf V}_{0\perp}^{\T}).
\end{equation}
\end{lemma}

\begin{lemma}\label{lem:svd}
The matrix $\bf B$ with singular value decomposition ${\bf B} = {\bf U} \text{Diag}({\bf S}) {\bf V}^{\T}$( with strictly positive singular value s) is optimal for the problem in \eqref{POLS} if and only if
\[
\hat{{\bf \Sigma}}_{xx}{\bf B} - \hat{{\bf \Sigma}}_{{\bf X}y} + \omega_n {\bf UV}^{\T} + N = 0,
\]
with ${\bf U}^{\T}{\bf N} = 0$, ${\bf NV} = 0$ and $\|{\bf N}\|_2 \leq \omega_n$.
\end{lemma}

\section{Proof of Theorems}\label{theoremproof}
\begin{proof}[Proof of Theorem \ref{tm1}]
Throughout the proof, we use $C$ to denote a universal positive constant where its value is not important for the theoretical purpose. 
In order to apply Lemma~\ref{lemma:nega}, we just need to evaluate the term $\|n^{-1} \sum_{i=1}^n \epsilon_i {\bf X}_i\|_2$ and then set the tuning parameter $w_n$ to be greater than that quantity. Note that $\epsilon_i = Y_i - \langle {\bf X}_i, {\bf B} \rangle - \beta_0^*$. Let ${\bf X}_i = \pi_1 {\bf X}_i^{(1)} + \pi_2 {\bf X}_i^{(2)}$, where vec$({\bf X}_i^{(g)}) \stackrel{i.i.d.}{\sim} N(\mu_g,{\bf \Sigma})$ and $\mu_g \in \mathbb{R}^{p q \times 1}$ for $g=1,2$. Define ${\bf X}_i \stackrel{\text{i.i.d}}{\sim} {\bf X}$ and $\epsilon_i \stackrel{\text{i.i.d}}{\sim} \epsilon$.  Observe that
\begin{align}\label{calmean}
&  \text{vec}\{\text{E}(\epsilon_i {\bf X}_i) \} \nonumber  \\
&= \pi_1 \text{E}\{(- \frac{n}{n_1} - \beta_0^* - \langle {\bf X}^{(1)}, {\bf B}_0 \rangle )\text{vec}({\bf X}^{(1)})  \}+ \pi_2 \text{E}\{(\frac{n}{n_2} - \beta_0^* - \langle {\bf X}^{(2)}, {\bf B}_0 \rangle) \text{vec}({\bf X}^{(2)})\}\nonumber  \\
 & = (\mu_2 - \mu_1) -(\pi_1\mu_1 +\pi_2\mu_2)\beta_0^* - \pi_1E\{\text{vec}({\bf X}^{(1)})\text{vec}({\bf X}^{(1)})^{\T}\}\text{vec}({\bf B}_0) \nonumber \\
 & ~~~~~~\quad \quad -  \pi_2E\{\text{vec}({\bf X}^{(2)})\text{vec}({\bf X}^{(2)})^{\T}\}\text{vec}({\bf B}_0) \nonumber  \\
 & = (\mu_2 - \mu_1) -(\pi_1\mu_1 +\pi_2\mu_2  )\beta_0^* - \pi_1  \{\mu_1 \mu_1^{\T} + {\bf \Sigma}\}\text{vec}({\bf B}_0) -  \pi_2 \{\mu_2 \mu_2^{\T} +{\bf \Sigma} \}\text{vec}({\bf B}_0).
\end{align}
Now, to further simplify this result, we reparameterize the mean of two normal populations such that $\mu_1 = 0$, and $\mu_2 = {\bf D}$. Then recall by the equivalence between LDA and least squares solution, we have
\begin{align*}
 \text{vec}({\bf B}) = c {\bf \Sigma}^{-1} {\bf D},~~ \beta_0 = -(\pi_1\mu_1 + \pi_2\mu_2)^{\T} \text{vec}({\bf B}) = - \pi_2 c {\bf D}^{\T} {\bf \Sigma}^{-1} {\bf D},~~\beta_0^* = \beta_0 - d
\end{align*}
for some positive constants $c$ and $d$. Then \eqref{calmean} can be simplified into
\begin{align*}
 & {\bf D} - \pi_2 {\bf D}   \beta_0^* -    \pi_2 \{{\bf D} {\bf D}^{\T}  \}\text{vec}({\bf B})  - c{\bf D} \\
& = {\bf D} - \pi_2 {\bf D} \beta_0 + \pi_2 {\bf D} {\bf D}  - \pi_2 \{{\bf D} {\bf D}^{\T}  \}\text{vec}({\bf B})  - c{\bf D} \\
& = {\bf D}\{ 1 + \pi_2^2 c {\bf D}^{\T} {\bf \Sigma}^{-1}{\bf D} +\pi_2 d - \pi_2 c {\bf D}^{\T} {\bf \Sigma}^{-1} {\bf D} - c\} \\
& = 0,
\end{align*}
given $d$ is chosen as $\pi_2^{-1} \{ c - 1+(\pi_2 - \pi_2^2) ({\bf D}^{\T} {\bf \Sigma}^{-1}{\bf D})\}$. 
%By Condition (A4), we have $\| \text{E}(\epsilon  X )\|_2 \leq c C_{\mu} (\sqrt{p} + \sqrt{q})$.  
%\textcolor{red}{this quantity is always non-zero... Therefore the OLS estimate is actually inconsistent, because $E(\epsilon X) \neq 0$. Think about taking derivative of $\sum_{i=1}^n (Y_i - \beta_0 - X \beta)^2$ with respect to $\beta$, one would expect to obtain $E (\epsilon X) = 0$. The main reason is that by taking the derivative of squared loss function with respect to $\beta_0$ and set it to $0$. We essentially require $E (\epsilon) = 0$. This is how we obtain the intercept $\beta_0$, however, it is known that the solution to OLS is only equivalent with LDA's solution in terms of the slope $B$, not on $\beta_0$. Therefore we consider redefine the intercept as $\beta_0^* = \beta_0 - d$ with  a constant $d$. Then the residual $\epsilon$ is going to shift by $d$ as well. Then note that the new $\text{vec}E(\epsilon X) =\left\{(\pi_2^2 - \pi_2)D D^{\T} - \Sigma  \right\} \text{vec}(B)   + d E(X) = (\pi_2^2 - \pi_2)D D^{\T}c \Sigma^{-1}D  - cD + d (\pi_1\mu_1 + \pi_2 \mu_2) =  (\pi_2^2 - \pi_2)D D^{\T}c \Sigma^{-1}D  - cD + d ( \mu_1 + \pi_2 D) $. Now we reparameterize the mean of two normal populations such that $\mu_1 = 0$, and $\mu_2 = D$. Then we can always find a constant $d= \pi_2^{-1}c (1+(\pi_2 - \pi_2^2) (D^{\T} \Sigma^{-1}D)$ such that the $\text{vec}E(\epsilon X) = c (\pi_2^2 - \pi_2)D (D^{\T}  \Sigma^{-1}D)  - cD + d  \pi_2 D$ becomes $0$. }

Next we show that with high probability, $\|\epsilon  {\bf X} \|_2 \leq  2 \log n (C_{\mu} + \lambda_u^{1/2}) (\sqrt{p} + \sqrt{q} +   \sqrt{\log n})  $. Since $\epsilon$ follows a mixture of two normal distributions, $\epsilon$ is sub-gaussian with sub-gaussian parameter denoted by $\sigma$, which is a positive constant due to the bounded eigenvalue condition in (A1). By Lemma \ref{lemma:Hoeffding}, for sufficiently large $n$,
\[
P(|\epsilon| > 2\log n) \leq P(|\epsilon - E(\epsilon)| > \log n)\leq 2\exp(-\frac{\log^2 n}{2\sigma^2})\leq C \exp(-2\log n) = \frac{C}{n^2}.
\]
Then we know $|\epsilon|\leq  2\log n$ with probability of at least $1 - C n^{-2}$. For $\|{\bf X}\|_2$, we first consider its centralized version, that is, ${\bf X} \sim N(0,{\bf\Sigma})$. Note that we can write the spectral norm of a matrix in the form of a canonical Gaussian process, 
\begin{align*}
\|N(0,  {\bf \Sigma})\|_2 = \sup_{{\bf A}: \|{\bf A}\|_\ast\leq 1} \langle  N(0,  {\bf \Sigma}), {\bf A} \rangle.
\end{align*}
This allows us to apply Gaussian comparison inequality \citep{Slepian1962}. Define ${\bf Z} \in \mathbb{R}^{p \times q}$ that satisfies vec$({\bf Z}) \sim N(0,{\bf I})$. Then by Lemma \ref{lemma:comparison}, we have
\begin{align}
 P(\|N(0,  {\bf \Sigma})\|_2 >  t_1) &= P\Big(\sup_{A: \|A\|_\ast\leq 1}  \langle N(0,  {\bf \Sigma}) , {\bf A}\rangle >  t_1\Big) \nonumber \\
&\leq  P\Big(\sup_{{\bf A}: \|{\bf A}\|_\ast \leq 1} \langle {\bf Z}, {\bf A}\rangle > t_1 \lambda_u^{-1/2}  \Big) \nonumber \\
& = P(\|{\bf Z}\|_2 >  t_1 \lambda_u^{-1/2})  \label{eeee}
\end{align}
for any $t_1 > 0$ because ${\bf \Sigma} \leq \lambda_u {\bf I}$ due to (A1). Apply Lemma \ref{lemma:concentration} (or more generally the Tracy-Widow law), we have 
\begin{align*}
  P(\|{\bf Z}\|_2 - E \|{\bf Z}\|_2 >  \sqrt{\log n}  )  \leq C \exp(- 2 \log n) = C n^{-2}
\end{align*}
for some constant $C > 0$. Since $E \|{\bf Z}\|_2  \leq \sqrt{p} + \sqrt{q}$, by Lemma \ref{lemma:gordon}, with probability of at least $ 1 - Cn^{-2}$,  $
   \|{\bf Z}\|_2 \leq \sqrt{p} + \sqrt{q} +  \sqrt{\log n}$, which leads to $\|N(0,{\bf \Sigma})\|_2 \leq \lambda_u^{1/2} (\sqrt{p} + \sqrt{q} + \sqrt{\log n})$ by \eqref{eeee}. Therefore with probability of at least $1 - C n^{-2}$, 
\begin{align*}
\|\epsilon  {\bf X}  \|_2 &  \leq (2 \log n) \|{\bf X}\|_2 \\
& \leq 2 \log n (\|\mu_1\|_2 + \|N(0,{\bf \Sigma})\|_2) \\
& \leq 2 \log n\left\{C_{\mu} (\sqrt{p} + \sqrt{q}) + \lambda_u^{1/2} (\sqrt{p} + \sqrt{q} + \sqrt{\log n})\right\} \\
& \leq  2 \log n (C_{\mu} + \lambda_u^{1/2}) (\sqrt{p} + \sqrt{q} +   \sqrt{\log n}) 
\end{align*}
using Condition (A4) and since we assume $\mu_2 = 0$ without loss of generality. 

Now we apply the standard matrix concentration inequality, (e.g., Lemma \ref{lemma:matrix-con}) with $M = 2 \log n (C_{\mu} + \lambda_u^{1/2}) (\sqrt{p} + \sqrt{q} +  \sqrt{\log n}) $ and $\delta = n^{-1}$. Note that $P(\|{\bf X}_i \epsilon_i \|_2 \leq M,~~i=1,\ldots,n) = (1 - Cn^{-2})^n \geq 1 - Cn^{-1}$ by Bernoulli's inequality. Hence we obtain that with probability of at least $ 1- C n^{-1}$, 
\begin{align*}
\left\|\frac{1}{n} \sum_{i=1}^n {\bf X}_i \epsilon_i  - \text{E}(\epsilon {\bf X} )  \right\|_2 &
\leq \frac{6M}{\sqrt{n}} \left( \sqrt{\log \min(p,q)} + \sqrt{\log 1/\delta}\right) \\
& \leq \frac{12 (\log n)^{3/2} (C_{\mu} + \lambda_u^{1/2}) (\sqrt{p} + \sqrt{q} +  \sqrt{\log n}) }{\sqrt{n}} 
\end{align*}
 This completes the proof. 
\end{proof}

\begin{proof}[Proof of Theorem \ref{thm2}]
By Lemma \ref{prop1}, we obtain $\hat {\bf B} = {\bf B}_0 + \omega_n{\bf \Delta} + o_p(\omega_n)$. Since the rank is a lower semi-continuous function, the rank of $\hat {\bf B}$ is larger than $r$ with probability tending to one by the consistency result, where $r$ is the rank of ${\bf B}_0$. Suppose $ \hat {\bf B}$ has singular value decomposition $USV^{\T}$ and $U_c$, $V_c$ are singular vectors corresponding to $U$ and $V$ except the $r$ largest singular values. By Lemma~\ref{lem:svd}, $\hat{\bf \Sigma}_{xx}(\hat {\bf B} - {\bf B}_0) - \hat{\bf \Sigma}_{{\bf X}\epsilon}$ and $\hat {\bf B}$ have simultaneous singular value decomposition. Therefore it suffices to show$ \|{\bf U}_c^{\T}\{\hat{\bf \Sigma}_{xx}(\hat {\bf B} - {\bf B}_0) - \hat{\bf \Sigma}_{{\bf X}\epsilon}\}{\bf V}_c\|_2 < \omega_n$ with probability tending to one.
Note that \
\[
\begin{split}
{\bf U}_c^{\T}\{\hat{\bf \Sigma}_{xx}(\hat {\bf B} - {\bf B}_0) - \hat{\bf \Sigma}_{{\bf X}\epsilon}\}{\bf V}_c
& = {\bf U}_c^{\T}\{\omega_n\hat{\bf \Sigma}_{xx}{\bf \Delta} + o_p(\omega_n) - O_p(n^{-1/2})\}{\bf V}_c\\
& = \omega_n {\bf U}_c^{\T}({\bf \Sigma} {\bf \Delta}){\bf V}_c + o_p(\omega_n),
\end{split}
\]
 where ${\bf \Sigma} {\bf \Delta}$ is the matrix in $\mathrm{R}^{p\times q}$ satisfying $\text{vec}({\bf \Sigma} {\bf \Delta}) = {\bf \Sigma} \text{vec}({\bf \Delta})$.
Because of the regular consistency and a positive eigengap for ${\bf B}_0$, the projection onto the first singular vectors of $\hat {\bf B}$ converges those of ${\bf B}_0$. Hence the projection on the orthogonal space is also consistent, which means ${\bf U}_c{\bf U}_c^{\T}$converges to ${\bf U}_{0\perp}{\bf U}_{0\perp}^{\T}$ and ${\bf V}_c{\bf V}_c^{\T}$converges to ${\bf V}_{0\perp}{\bf V}_{0\perp}^{\T}$. Then by Lemma~\ref{lem:sic}, we have
\[
\begin{split}
\|{\bf U}_c^{\T}\{\hat{\bf \Sigma}_{xx}(\hat {\bf B} - {\bf B}_0) - \hat{\bf \Sigma}_{{\bf X}\epsilon}\}{\bf V}_c\|_2 
&= \|{\bf U}_c{\bf U}_c^{\T}\{\hat{\bf \Sigma}_{xx}(\hat {\bf B} - {\bf B}_0) - \hat{\bf \Sigma}_{{\bf X}\epsilon}\}{\bf V}_c{\bf V}_c^{\T}\|_2\\
&= \omega_n \|{\bf U}_{0\perp}{\bf U}_{0\perp}^{\T}({\bf \Sigma} {\bf \Delta}){\bf V}_{0\perp}{\bf V}_{0\perp}^{\T}\|_2 + o_p(\omega_n)\\
& =  \omega_n \|{\bf U}_{0\perp}{\bf U}_{0\perp}^{\T}{\bf \Sigma} \{-{\bf \Sigma}^{-1}({{\bf U}_0{\bf V}_0}^{\T} -{\bf U}_{0\perp}\Lambda {\bf V}_{0\perp}^{\T})\}{\bf V}_{0\perp}{\bf V}_{0\perp}^{\T}\|_2 + o_p(\omega_n)\\
& = \omega_n\|{\bf U}_{0\perp}{\bf \Lambda} {\bf V}_{0\perp}^{\T}\|_2 + o_p(\omega_n)\\
& = \omega_n\|{\bf \Lambda}\|_2 + o_p(\omega_n),
\end{split}
\]
where the third equality is due to \eqref{deltasolution}.
Since $\|{\bf \Lambda}\|_2 < 1$, the last expression is less than $\omega_n$ with probability tending to one, which completes the proof.
\end{proof}
\begin{proof}[Proof of Theorem \ref{tm3}]
Based on Corollary 3.1 of \citet{Zhang2004}, we have $$R(\hat{f}_n) \leq R^* + 2 c(\epsilon_1 + \epsilon_2)^{1/s},$$ where $Q$ is the squared error loss function defined by $Q(f) = E_{\bf X} \{y - f({\bf X})^2\}$,  $\epsilon_1 = \inf_f E_{\bf X} (2P(Y=1 \mid {\bf X}) - 1 - f({\bf X}))^2$, $\epsilon_2$ satisfies $Q(\hat{f}_n) \leq \inf_f Q(f) + \epsilon_2$, and $c$ and $s$ can be chosen as $c=0.5$ and $s=2$ as explained by the Example 3.1 (for least squares loss function) in that paper. Now note that since $\hat{f}_n$ is determined by the classification coefficient $\hat{\bf B}$ and $\hat{\beta}_0$ that are both consistent based on Theorem \ref{tm1}. Therefore $\epsilon_2$ can be chosen arbitrarily close to $0$. Also, as we assume the true class label $Y$ given ${\bf X}$ is determined by the linear classification rule with $\beta_0^*$ and ${\bf B}_0$, then $\inf_f E_{\bf X} \{2P(Y=1 \mid {\bf X}) - 1 - f({\bf X})\}^2 = 0$. Therefore $\epsilon_1 = 0$. This concludes the proof. 
%\citet{Bartlett2006} 
\end{proof}

\bibliographystyle{plainnat}

\bibliography{paper-ref}
\end{document}